\documentclass[10pt,journal,letterpaper,twocolumn,twoside,nofonttune]{IEEEtranTCOM}
\usepackage{etex}
\normalsize
\usepackage[T1]{fontenc}
\usepackage{amsmath,amssymb,amsfonts}
\usepackage{mathrsfs}
\usepackage{mathabx}
\usepackage{amsbsy}
\usepackage{graphicx}
\usepackage{graphics}
\usepackage{epstopdf}
\usepackage{algorithm}
\usepackage{algorithmic}
\usepackage{subfigure}
\usepackage{cite}
\usepackage{slashbox}
\usepackage{multirow}
\usepackage{ctable}
\usepackage{tikz,pgfplots}
\usepackage{caption}

\title{\Huge 
Achieving the Uniform Rate Region of General Multiple Access Channels by Polar Coding\\[0.50ex]}

\renewcommand{\markboth}[2]
{\renewcommand{\leftmark}{#1}\renewcommand{\rightmark}{#2}}

\author{\large%
Hessam Mahdavifar, 
        Mostafa El-Khamy, 
        Jungwon Lee,
        Inyup Kang\\[2.25ex]
\vspace{-3.75ex}
\thanks{%
   Hessam Mahdavifar, Mostafa El-Khamy, Junwon Lee and Inyup Kang are with Modem R\&D, Samsung Electronics, San Diego, CA 92121, U.S.A.
   (e-mail: \{h.mahdavifar,\,mostafa.e,\,jungwon2.lee,\,inyup.kang\}@samsung.com).
}}

\newtheorem{theorem}{{Theorem}}
\newtheorem{lemma}[theorem]{{Lemma}}

\newtheorem{corollary}[theorem]{{Corollary}}

\newcommand{\cA}{{\cal A}} 
\newcommand{\cB}{{\cal B}}
 
\newcommand{\cD}{{\cal D}}

\newcommand{\cG}{{\cal G}}

\newcommand{\cJ}{{\cal J}}

\newcommand{\cM}{{\cal M}}

\newcommand{\cP}{{\cal P}}

\DeclareMathAlphabet{\mathbfsl}{OT1}{ppl}{b}{it} 

\newcommand{\bQ}{\mathbfsl{Q}} 
\newcommand{\bR}{\mathbfsl{R}}

\newcommand{\bX}{\mathbfsl{X}}

\newcommand{\be}[1]{\begin{equation}\label{#1}}
\newcommand{\ee}{\end{equation}} 
\newcommand{\eq}[1]{(\ref{#1})}

\renewcommand{\leq}{\leqslant}
 
\renewcommand{\geq}{\geqslant}

\newcommand{\script}[1]{{\mathscr #1}}

\renewcommand{\Bbb}{\mathbb}
 
\newcommand{\N}{{\Bbb N}}
\newcommand{\R}{{\Bbb R}}

\newcommand{\Tref}[1]{Theo\-rem\,\ref{#1}}

\newcommand{\Lref}[1]{Lem\-ma\,\ref{#1}}
\newcommand{\Cref}[1]{Co\-ro\-lla\-ry\,\ref{#1}}

\newcommand{\deff}{\mbox{$\stackrel{\rm def}{=}$}}

\newcommand{\Gn}{G^{\otimes n}}

\newcommand{\sX}{\script{X}}
\newcommand{\sY}{\script{Y}}

\newcommand{\sI}{\script{I}}

\newcommand{\shalf}{\mbox{\raisebox{.8mm}{\footnotesize $\scriptstyle 1$}
\footnotesize$\!\!\! / \!\!\!$ \raisebox{-.8mm}{\footnotesize
$\scriptstyle 2$}}}

\newcommand{\AQ}{\cA_{\bQ}^{(\pi)}}
\newcommand{\rp}{\bR^{(\pi)}}

\begin{document}

\maketitle

\begin{abstract}

We consider the problem of polar coding for transmission over $m$-user multiple access channels. In the proposed scheme, all users encode their messages using a polar encoder, while a joint successive cancellation decoder is deployed at the receiver. The encoding is done separately across the users and is independent of the target achievable rate, in the sense that the encoder core is the regular Ar{\i}kan's polarization matrix. For the code construction, the positions of information bits and frozen bits for each of the users are decided jointly. This is done by treating the whole polar transformation across all the $m$ users as a single polar transformation with a certain \emph{base code}. We prove that the \emph{covering radius} of the dominant face of the uniform rate region is upper bounded by $r = \frac{(m-1)\sqrt{m}}{L}$, where $L$ represents the length of the base code. We then prove that the proposed polar coding scheme achieves the whole uniform rate region, with small enough resolution characterized by $r$, by changing the decoding order in the joint successive cancellation decoder. The encoding and decoding complexities are $O(N \log N)$, where $N$ is the code block length, and the asymptotic block error probability of $O(2^{-N^{0.5 - \epsilon}})$ is guaranteed. Examples of achievable rates for the case of $3$-user multiple access channel are provided. 

\end{abstract}


\section{Introduction} 
\label{sec:Introduction}

\noindent 
\PARstart{P}{ol}ar codes were introduced by Ar{\i}kan in the seminal work of \cite{Arikan}. They are the first family of codes for the class of binary-input symmetric discrete memoryless channels that are provable to be capacity-achieving with low encoding and decoding complexity. Construction of polar codes is based on a phenomenon called the \emph{channel polarization}. Ar{\i}kan proves that as the block length goes to infinity the channels seen by individual bits through a certain transformation called the \emph{polar transformation} start polarizing which means that they approach either a noise-less channel or a pure-noise channel. In general, the constructed polar codes, together with the low complex \emph{successive cancellation decoder}, achieve the \emph{symmetric capacity} of all binary-input memoryless channels, where the symmetric capacity of a binary-input channel is the mutual information between the input and output of the channel assuming that the input distribution is uniform. 

Polar codes and polarization phenomenon have been successfully applied to various problems such as wiretap channels \cite{MV}, data compression \cite{Arikan2,Ab} and multiple access channels \cite{STY,AT2}. The notion of channel polarization has been extended to two user multiple access channels (MAC) \cite{STY} and later to $m$-user MAC \cite{AT2}, wherein a technique is described to polarize a given binary-input MAC same as in Ar{\i}kan's groundbreaking work of \cite{Arikan}. The authors convert multiple uses of this MAC into single uses of extremal MACs and characterize the set of all extremal MACs in the asymptotic sense. 

The capacity region of multiple access channels is fully characterized by Ahlswede \cite{A} and Liao \cite{L} for the case that the sources transmit independent messages. However, in this paper, we are only interested in the \emph{uniform rate region}. For a multiple access channel, the uniform rate region is the achievable region corresponding to the case that the input distributions are uniform. The single user counterpart of uniform rate region is indeed symmetric capacity. It has been shown that at least one point on the \emph{dominant face} of the uniform rate region can be achieved by the polar code constructed based on MAC polarization \cite{STY,AT2}. But the problem of achieving the entire uniform rate region for the general case of $m$-user MAC remained open.

The core idea behind the approach of this paper, which distinguishes it from the approach taken in the prior works of \cite{STY,AT2} can be explained as follows. It is well-known that the \emph{corner points} of the uniform rate region are achievable with capacity-achieving single user codes. We observe that the successive method of decoding transmitted messages of different users at corner points follows the same logic as in the successive cancellation decoder of polar codes. Therefore, if polar encoders are deployed by the users, then the receiver can implement a joint successive cancellation decoder in such a way that all the messages are decoded in a unified manner. It can be shown that any particular decoding order will result in a certain achievable set of rates. Then in order to achieve different rates, the idea is to shuffle not only the encoded messages but also the encoded bits in different messages. The decoding order will be determined a priori and will be known at both the transmitters and the receiver. We show the direct relation between the decoding order and the polarization base code which guarantees the channel polarization. Furthermore, we show how to change the decoding order to approach all the points in the uniform rate region using the single user channel polarization theorem.

In another related work, Ar{\i}kan proposed a scheme for the Slepian-Wolf source coding problem, which is the dual of the special case of two-user MAC, based on \emph{monotone} chain rule expansions and it is shown that the uniform rate region is achievable with polar coding \cite{Arikan3}. In \cite{arxiv,CISS}, we showed how the polar coding with joint successive cancellation decoding can be used to achieve the uniform rate region of two-user multiple access channels, by regarding the MAC as one level of polarization, along with methods to improve the finite length performance. A method for improving the performance of two-user MAC polar coding with list decoding has been described in \cite{S}. A straightforward generalization of these methods for the $m$-user case provides achievability of the one-dimensional edges of the dominant face region. However, the dominant face of the uniform rate region is in general a polytope in an $m-1$-dimensional plane in the $m$-dimensional space and the main question was whether one can approach all points on the dominant face of the uniform rate region by changing the decoding order, thereby establishing the achievability of the entire region.  

As opposed to the approaches taken in \cite{Arikan3,S,CISS}, we do not limit ourselves to a certain set of polarization base codes. In this paper, the set of all possible permutations are considered for the decoding order, assuming that the decoding order within the input bits of each user is preserved. Any polarization base code of length $L$ leads to an $m$-tuple rate, which is proved to be achievable with polar coding. We further prove that the set of these achievable $m$-tuple rates cover the entire dominant face assuming $m$-dimensional balls of radius $\frac{(m-1)\sqrt{m}}{L}$ centered around them. This shows that all the points on the dominant face of the uniform rate region can be approached as the code block length and the length of polarization base codes grow large. 

The rest of this paper is organized as follows. In Section\,\ref{sec:two}, we review some background on polar codes and multiple access channels. In Section\,\ref{sec:five}, the special case of $3$-user MAC is considered and the set of all achievable rates for base codes of length $2$ are derived for an example. In Section\,\ref{sec:three}, we introduce the general MAC polarization base codes and prove that their corresponding achievable rates cover the entire dominant face of the uniform rate region. In Section\,\ref{sec:four}, we discuss how general MAC polar transformations can be built upon polarization base codes and establish the channel polarization theory. We also discuss the decoding method and provide comparison with prior work on MAC polar coding. At the end, we conclude the paper in Section\,\ref{sec:six}.

\section{Preliminaries}
\label{sec:two}

\subsection{Polar codes}

In this subsection, we provide a brief overview of the groundbreaking work of
Ar{\i}kan~\cite{Arikan} and others~\cite{AT,Korada,KSU} on polar 
codes and channel polarization. 

Polar codes are constructed based upon a phenomenon called \emph{channel polarization} discovered by Ar{\i}kan \cite{Arikan}. 
The basic polarization matrix is given as
\be{G-def}
G 
\ = \
\left[ 
\begin{array}{cc}
1 & 0\\
1 & 1\\ 
\end{array}
\right]
\ee

Consider two independent copies of a binary-input discrete memoryless channel (B-DMC) $W: \{0,1\} \to \sY$. The two input bits $(u_1,u_2)$, drawn from independent uniform distributions, are multiplied by $G$ and then transmitted over the two copies of $W$. One level of channel polarization is the mapping $(W,W) \rightarrow (W^{-}, W^{+})$, where 
$W^-: \{0,1\} \to \sY^2$, and  $W^+: \{0,1\} \to \{0,1\} \times \sY^2$ with the following channel transformation 
\begin{equation}
\begin{split}
\label{channel-com}
W \boxcoasterisk W(y_1,y_2|u_1) &= \frac{1}{2}  \sum_{u_2 \in \{0,1\}} W(y_1 | u_1 \oplus u_2) W(y_2|u_2), \\
W \circledast W(y_1,y_2,u_1|u_2) &=  \frac{1}{2} W(y_1 | u_1 \oplus u_2) W(y_2|u_2).
\end{split}
\end{equation}
$W \boxcoasterisk W$ and $W \circledast W$ are also denoted by $W^+$ and $W^-$. The \emph{bit-channels} $W^-$ and $W^+$ are indeed the channels that $u_1$ and $u_2$ observe assuming the following scenario: the first bit $u_1$ is decoded assuming $u_2$ is noise, then $u_2$ is decoded assuming that $u_1$ is decoded successfully and is known.

The channel polarization is continued recursively by further splitting $W^{-}$ and $W^{+}$ to get $W^{--}$, $W^{-+}$, $W^{+-}$, $W^{++}$ etc. This process can be explained best by means of Kronecker powers of $G$. The Kronecker powers of $G$ are defined by induction. Let $G^{\otimes 1} = G$ and for any $n > 1$:
$$
G^{\otimes (n)}
\ = \
\left[ 
\begin{array}{cc}
G^{\otimes (n-1)} & 0\\
G^{\otimes (n-1)} & G^{\otimes (n-1)}\\ 
\end{array}
\right]
$$
It can be observed that $G^{\otimes (n)}$ is a $2^n \times 2^n$ matrix. Let $N = 2^n$. Then $\Gn$ is the $N \times N$ polarization matrix. Let $(U_1, U_2,\dots,U_N)$, denoted by $U_1^N$, be a block of $N$ independent and uniform binary random variables. The polarization matrix $\Gn$ is applied to $U_1^N$ to get $X_1^N = U_1^N \Gn$. Then $X_i$'s are transmitted through $N$ independent copies of a binary-input discrete memoryless channel (B-DMC) $W$. The output is denoted by $Y_1^N$. This transformation with input $U_1^N$ and output $Y_1^N$ is called the polar transformation. In this transformation, $N$ independent uses of $W$ is transformed into $N$ bit-channels, described next. Following the convention, random variables are denoted by capital letters and their instances are denoted by small letters. Let $W^N: \sX^N \rightarrow \sY^N$ denote the channel consisting of $N$ independent copies of $W$ i.e. 
\be{Wn}
W^N\kern-0.5pt(y^N_1|x^N_1) 
\,\ \deff\,\
\prod_{i=1}^N W(y_i|x_i)
\vspace{-0.25ex}
\ee
The combined channel $\widetilde{W}$ is defined with transition probabilities given by
\be{Wtilde}
\widetilde{W}(y^N_1|u^N_1) 
\,\ \deff\,\
W^N\kern-1pt\bigl(y^N_1\hspace{1pt}{\bigm|}\hspace{1pt}u^N_1\hspace{1pt} \Gn\bigr)
\ee
For $i=1,2,\dots,N$, the bit-channel $W_N^{(i)}$ is defined as follows: 
\be{Wi-def}
W^{(i)}_N\bigl( y^N_1,u^{i-1}_1 | \hspace{1pt}u_i)
\,\ \deff \,\
\frac{1}{2^{n-1}}\hspace{-5pt}
\sum_{u_{i+1}^N \in \{0,1\}^{n-i}} \hspace{-12pt}
\widetilde{W}\Bigl(y^N_1\hspace{1pt}{\bigm|}\hspace{1pt}
u_1^N \Bigr)
\ee
Intuitively, this is the channel that bit $u_i$ observes through a \emph{successive cancellation decoder}, deployed at the output. Under this decoding method, proposed by Ar{\i}kan for polar codes \cite{Arikan}, all the bits $u_1^{i-1}$ are already decoded and are assumed to be available at the time that $u_i$ is being decoded. The channel polarization theorem  
states that as $N$ goes to infinity, the bit-channels start polarizing meaning that they either become a noise-less channel or a pure-noise channel. 

In order to measure how good a binary-input channel $W$ is, Ar{\i}kan uses the \emph{Bhattacharyya parameter} of $W$, denoted by $Z(W)$ \cite{Arikan}, defined as
$$ 
Z(W)
\,\ \deff\kern1pt
\sum_{y\in\sY} \!\sqrt{W(y|0)W(y|1)}
$$

It is easy to show that the Bhattacharyya parameter $Z(W)$ is always between $0$ and $1$. Channels with $Z(W)$ close to zero are almost noiseless, while channels with $Z(W)$ close to one are almost pure-noise channels. More precisely, it can be proved that the probability of error of a binary symmetric memoryless channel (BSM) is upperbounded by its Bhattacharyya parameter. Let $[N]$ denotes the set of positive integers less than or equal to $N$. The set of \emph{good bit-channels} $\cG_N(W,\beta)$ is defined for any $\beta < \shalf$ ~\cite{AT,Korada}:
\be{Arikan-good-def}
\cG_N(W,\beta) {\deff} \left\{\, i \in [N] ~:~ Z(W^{(i)}_{N}) < 2^{-N^{\beta}}\!\!/N \hspace{1pt}\right\}
\ee
Then the channel polarization theorem is proved by showing that the fraction of good bit-channels approaches the symmetric capacity $I(W)$, as $N$ goes to infinity \cite{AT}. This theorem readily leads to a construction of capacity-achieving \emph{polar codes}. The idea is to transmit the information bits over the good bit-channels while freezing the input to the other bit-channels to a priori known values, say zeros. The decoder for this constructed code is the successive cancellation decoder of Ar{\i}kan~\cite{Arikan}, where it is further proved that the frame error probability under successive cancellation decoding is upper bounded by the sum of Bhattacharyya parameters of the selected good bit-channels, which is $2^{-N^{\beta}}$ by the particular choice of good bit-channels in \eq{Arikan-good-def}. 

\subsection{Multiple access channel}

Let $W: \sX^m \rightarrow \sY$ be an $m$-user MAC, where $\sX$ is the binary alphabet and $\sY$ is the output alphabet. With slight abuse of notation, the channel is described by the transition probability $W(y|x[1],x[2],\dots,x[m])$ for any $x[j] \in \sX$, $j \in [m]$, and $y \in \sY$. Namely, $W(y|x[1],x[2],\dots,x[m])$ denote the probability of receiving $y$ given that $x[j]$ is transmitted by the $j$-th user user, for $j \in [m]$. For any $\cJ \subset [m]$, let $x[\cJ]$ denote the set $\left\{x[j]: j \in \cJ\right\}$. 

Let $X[1],X[2],\dots,X[m]$ be independent and uniform binary random variables generated by users $1,2,\dots,m$. Then the uniform rate region of $W$, denoted by $\sI(W)$, is defined to be the set of all $m$-tuples $\bR = (R_1,R_2,\dots,R_m) \in \R^m$ such that
\be{region-def}
0 \leq \sum_{j \in \cJ} R_j \leq I(X[\cJ];Y,X[\cJ^c]), \forall \cJ \subseteq [m].  
\ee

The uniform rate region is the set of all achievable $m$-tuples of rates assuming that the input distributions are uniform. Let $I(W)$, which is also called the \emph{uniform sum-rate} of $W$, be defined as follows:
$$
I(W) = I(X[1],X[2],\dots,X[m];Y)
$$ 
In general, in the context of this paper, any point $\bR = (R_1,R_2,\dots,R_m) \in \R^m$ is regarded as an $m$-tuple of rates and hence the $\sum R_j$,  is referred to as the sum-rate of $\bR$. 

The dominant face of the uniform rate region, denoted by $\cD(W)$, is defined to be the set of points in $\sI(W)$, with the maximum sum-rate $I(W)$ i.e. the right inequality of \eq{region-def} is in fact equality for $\cJ = [m]$. It can be observed that, the achievability of the uniform rate region $\sI(W)$ is equivalent to the achievability of its dominant face $\cD(W)$. Therefore, our focus throughout this paper is on the achievability of $\cD(W)$. 

  

\section{The proposed scheme for $3$-user Multiple Access Channels}
\label{sec:five}

In this section, we consider the special case of polar coding for $3$-user multiple access channels. Let $W$ be a $3$-user MAC with transition probability $W(y|x[1],x[2],x[3])$. The uniform rate region of $W$ is in general a $3$-dimensional polyhedron, which is shown for an example, described later, in Figure\,\ref{region2}, where $R_1$, $R_2$ and $R_3$ are rates of user $1$, $2$ and $3$, respectively. 

Let $X[1]$, $X[2]$ and $X[3]$ be uniform and independent binary random variables. Let $I_1$, $I_2$ and $I_3$ be defined as follows:
\begin{align*}
I_1 &= I(X[1];Y)\\
I_2 &= I(X[2];Y,X[1])\\
I_3 &= I(X[3];Y,X[1,2])
\end{align*}
The dominant face $\cD(W)$ is the hexagon whose vertices are specified in Figure\,\ref{face}. These vertices are also called \emph{corner points} of the uniform rate region. For simplicity, it is assumed that $W$ is symmetric with respect to the inputs i.e. $W(y|x[1],x[2],x[3])$ remains the same if $(x[1],x[2],x[3])$ is permuted.

\begin{figure}[h]
\centering
\includegraphics[width=\linewidth]{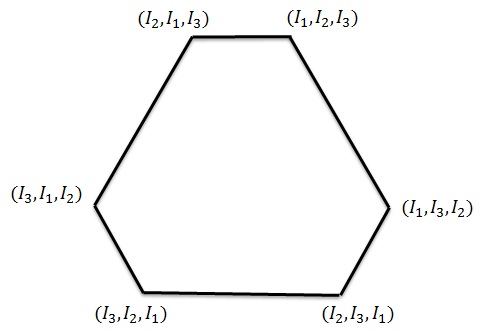}
\caption{The dominant face $\cD(W)$}
\label{face}
\end{figure} 

The corner points are achievable by \emph{separate} polar coding. For instance, in order to achieve the point $(I_1,I_2,I_3)$, the message of user $1$ is decoded first. Then the message of user $2$ is decoded assuming the user $1$'s message is known. At the end, the message of user $3$ is decoded, assuming the first and second messages are known. In this case, the scheme does not depend on the underlying polar codes and any other capacity achieving code will fit as well. We observed that the successive method of decoding the messages follows the same rule as in the successive cancellation decoding of polar codes. Therefore, we propose the joint successive cancellation decoding for all the users. The idea is to shuffle the decoding order of not only the messages, but also the bits within each message. In order to guarantee channel polarization, we let the decoding order for a certain polarization \emph{base code} to be arbitrarily chosen. Then the recursive steps of channel polarization are applied on top of the base code.

For the separate polar coding, the length of the base code is $1$, which is defined to be the length of the message for each user. There are $3! = 6$ different base codes which result in the achievability of $6$ corner points. The idea then is to increase the length of the base code and achieve more and more points. An example of a base code of length $2$ is shown in Figure\,\ref{building3}. The decoding order of this base code is given by $(x_1[1],x_1[2],x_1[3], x_2[1], x_2[2], x_2[3])$. This decoding order can be simply denoted by a permutation $\pi$ which permutes the default vector $(x_1[1],x_2[1],x_1[2], x_2[2], u_3[1], u_3[2])$. In this case, $\pi = (1,4,2,5,3,6)$. For polar coding with a general block length $N$ built upon this base code, the decoding order is specified as follows. The successive cancellation decoding decodes the first half of user $1$'s message first, then the first half of user $2$'s message and then the first half of user $3$'s message followed by the second half of the messages of user $1$, $2$ and $3$, respectively.

\begin{figure}[h]
\centering
\includegraphics[width=\linewidth]{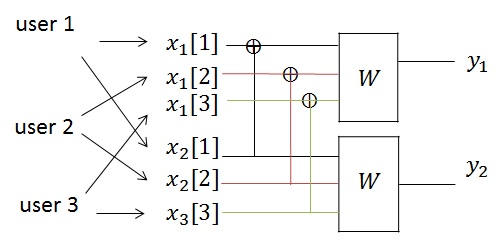}
\caption{An example for a base code of length $2$}
\label{building3}
\end{figure} 

In total, there are $\frac{6!}{2!2!2!} = 90$ possible base codes of length $2$. Notice that the base code has to keep the order within each user's message i.e. $x_1[j]$ has to appear before $x_2[j]$ in the decoding order. We have derived all these achievable points for an example. The system model assumes that the multiple access channel $W$ is a binary-additive Gaussian noise channel where inputs $x[1], x[2], x[3]  \in \left\{0,1\right\}$ are modulated using BPSK ($0$ is mapped to $-1$ and $1$ is mapped to $+1$) into $\overline{x}[1]$, $\overline{x}[2]$ and $\overline{x}[3]$, respectively. The output of the channel is denoted by $y$, where $y=\overline{x}[1]+\overline{x}[2]+\overline{x}[3]+N$ and $N$ is the Gaussian noise of unit variance $N_0$. For this channel, the capacity region is same as the uniform rate region and is given by the set of all possible $3$-tuple rates $(R_1,R_2,R_3)$ that satisfy
\begin{align*}
R_1,R_2,R_3 &\leq I_3  = 0.7215\\
R_1+R_2,R_1+R_3,R_2+R_3 &\leq I_3+I_2 = 1.1106\\
R_1+R_2+R_3  &\leq I_1+I_2+I_3 =  1.3681
\end{align*}

The uniform rate region along with all the $90$ achievable points on the dominant face with base codes of length $2$ are shown in Figure\,\ref{region2}. We prove in the next section that by letting the length of the base code to grow large, the achievable $3$-tuple rates of base codes cover the entire dominant face which will then be used to establish the achievability of the capacity region with polar coding.

\begin{figure}[h]
\centering
\includegraphics[width=\linewidth]{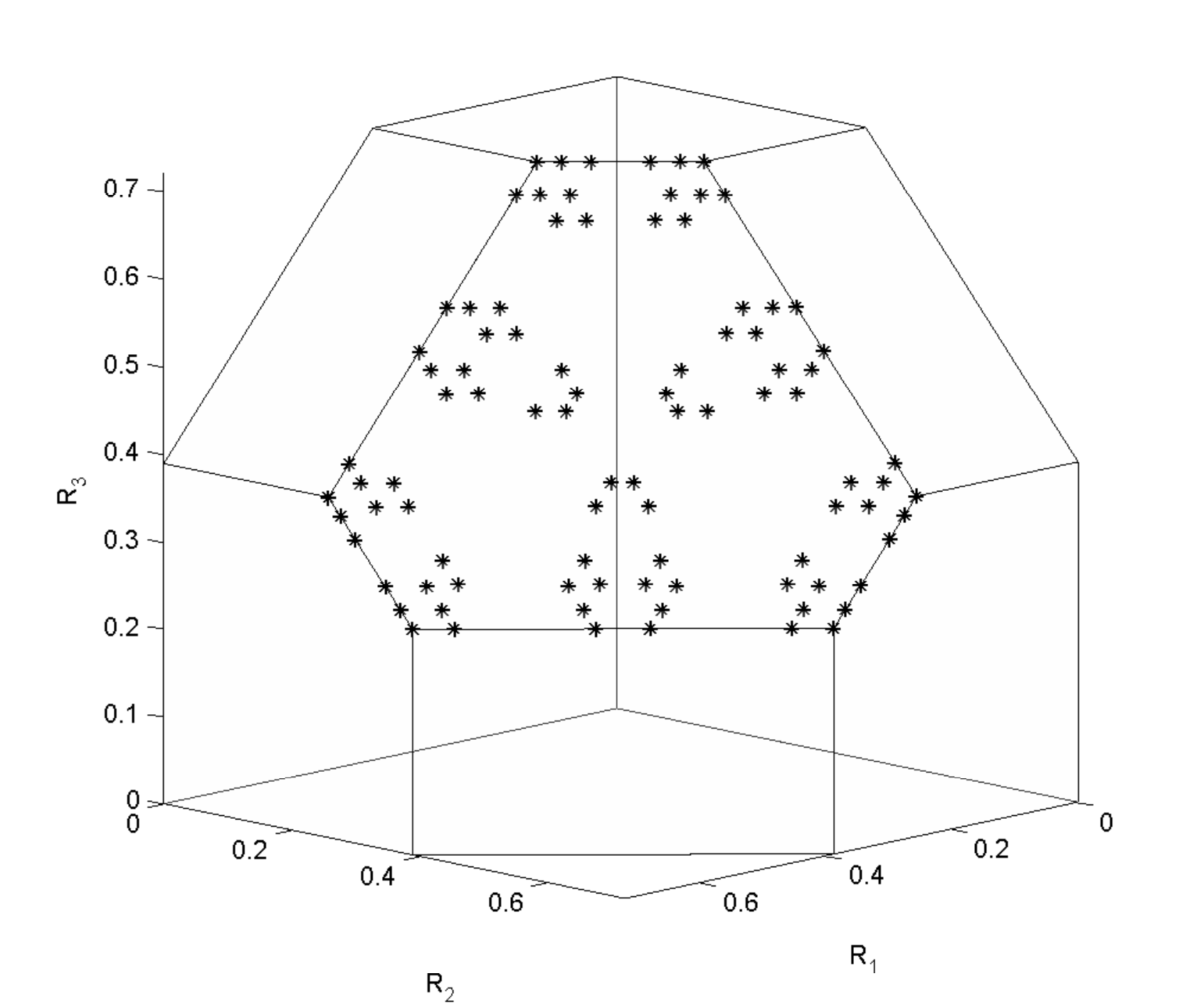}
\caption{The uniform rate region and the achievable points with base codes of length $2$}
\label{region2}
\end{figure}

\section{Covering radius of the dominant face}
\label{sec:three}

Let $W$ be a given $m$-user binary-input discrete multiple access channel. Let also $l \geq 1$ be a positive integer and $L = 2^l$. For $j = 1,2,\dots,m$, assume that $U_1^L[j]=(U_1[j],U_2[j],\dots,U_L[j])$ is a vector of independent and uniformly distributed bits that is generated by the $j$-th user, independent of the other users. We also refer to $U_1^L[j]$ as the input bits of the $j$-th user. Let also $X_1^L[j] = U_1^L[j]\,G^{\otimes l}$. For $i=1,2,\dots,L$, the $m$-tuple $(X_i[1],X_i[2],\dots,X_i[m])$ is transmitted through the $i$-th independent copy of $W$ and the output is denoted by $Y_i$. There are several ways to form an ordered sequence from the concatenated vector $\bigl(U_1^L[1],U_1^L[2],\dots,U_1^L[m]\bigr)$. For notational convenience, the concatenated vector $\bigl(U_1^L[1],U_1^L[2],\dots,U_1^L[m]\bigr)$ is also denoted by $D_1^{mL}$. In general, there are 
$$
{mL \choose L,L,...,L} = \frac{(mL)!}{L!^m}
$$
permutations $\pi$ assuming that the order of $U_1^L[j]$ is preserved through the permutation i.e. $U_i[j]$ appears in the permuted sequence before $U_k[j]$, for any $j$ and $i < k$. Let $\cP_L$ denote the set of all such permutations. The described MAC polar transformation along with the choice of permutation $\pi \in \cP_L$ represents a polarization base code of length $L$, on top of which the recursive channel polarization is applied. This procedure will be elaborated more in the next section. The permutation $\pi$ will enforce the decoding order in the joint successive cancellation decoder and hence, is also referred to as the decoding order.

The mutual information $I(D_1^{mL};Y_1^L)$ can be expanded with respect to the permutation $\pi$ using the chain rule as follows:
$$
I(D_1^{mL};Y_1^L) = \sum_{i=1}^{mL} I\bigl(D_{\pi^{-1}(i)};Y_1^L,D_{\pi^{-1}(1)}^{\pi^{-1}(i-1)}\bigr)
$$

Then for $i=1,2,\dots,mL$, let $I_i^{(\pi)}$ denote the $i$-th term in the above expansion i.e.
\be{Iji-dfn}
I_i^{(\pi)}\,\deff\, I\bigl(D_{\pi^{-1}(i)};Y_1^L,D_{\pi^{-1}(1)}^{\pi^{-1}(i-1)}\bigr)
\ee
Also, for $j = 1,2,\dots,m$, let
\be{Rj-dfn}
R_j^{(\pi)}\ \deff\ \frac{1}{L} \sum_{i=(j-1)L+1}^{jL} I_{\pi(i)}^{(\pi)}
\ee
In fact, intuitively speaking, $R_j^{(\pi)}$ is the allocated capacity to the $j$-th user. As we will see in the next section, the $m$-tuple of rates $\bR^{(\pi)} = (R_1^{(\pi)},R_2^{(\pi)},\dots,R_m^{(\pi)})$ is achievable with polar coding built upon the base code with permutation $\pi$. But before that, through the rest of this section, we characterize the \emph{covering radius} of the dominant face. The covering radius $r$, with respect to the length of the base codes $L$, is formally defined as follows:
\be{covering-radius}
r \,\deff\, \max_{\bQ \in \cD(W)} \min_{\pi \in \cP_L} \left\| \bQ - \bR^{(\pi)}\right\|
\ee
where $\left\|\bX\right\|$ is the Euclidean norm of $X$ in the $m$-dimensional space $\R^m$.

The following lemma is the result of \eq{Iji-dfn} and \eq{Rj-dfn}.
\begin{lemma}
\label{sum-rate}
For any permutation $\pi$, the sum-rate of $\bR^{(\pi)}$ is $I(W)$.
\end{lemma}
\noindent
{\bf Remark.} \looseness=-1 
In fact a more general statement than \Lref{sum-rate} holds. For any permutation $\pi$, $\bR^{(\pi)}$ is a point on the dominant face $\cD(W)$. This will also follow as a result of the next section, where we prove that these points are achievable by polar coding.

For two users $j_1$ and $j_2$, we write $j_1 \rightarrow j_2$ if an input bit of the user $j_1$ appears right before an input bit of the user $j_2$ through the permutation $\pi$. More precisely, there exist $i_1$ and $i_2$ with $(j_1-1)L+1 \leq i_1 \leq j_1 L$ and $(j_2-1)L+1 \leq i_2 \leq j_2 L$ such that $\pi(i_2) = \pi(i_1)+1$. In that case, we also define the new permutation $\pi'$ from $\pi$ by swapping $\pi(i_1)$ and $\pi(i_2)$ i.e. $\pi'(i_2) = \pi(i_1)$, $\pi'(i_1) = \pi(i_2)$ and $\pi'(i) = \pi(i)$ for $i \neq i_1,i_2$. We say that $\pi'$ is the transposition of $\pi$ with respect to $j_1 \rightarrow j_2$. Notice that the transposition with respect to $j_1 \rightarrow j_2$ is not necessarily unique, as there may be other input bits of the user $j_1$ that appear right before other input bits of the user $j_2$. In that case, we may choose one of them for the transposition.

\begin{lemma}
\label{resolution1}
Let $\pi'$ be the transposition of $\pi$ with respect to $j_1 \rightarrow j_2$. Then we have
$$
R_j^{(\pi)} = R_j^{(\pi')}\ \text{for}\ j \neq j_1,j_2 
$$
and
$$
0 \leq R_{j_1}^{(\pi')} - R_{j_1}^{(\pi)} = R_{j_2}^{(\pi)} - R_{j_2}^{(\pi')} \leq \frac{1}{L}
$$
\end{lemma}
\begin{proof}
By definition, $\pi'$ is the result of $\pi$ by swapping $\pi(i_1)$ and $\pi(i_2)$, where $\pi(i_2) = \pi(i_1) +1$. Therefore, by definition of $I_i^{(\pi)}$ in \eq{Iji-dfn}, $I_i^{(\pi)} = I_i^{\pi'}$ for $i \neq i_1,i_2$. This implies that the values of $R_j^{\pi}$ do not change during the specified transposition, for $j \neq j_1,j_2$. 

For the second part, 
\begin{equation}
\begin{split}
\label{resolution-eq1}
L(R_{j_1}^{(\pi')} - R_{j_1}^{(\pi)}) &= \frac{\sum_{i=(j_1-1)L+1}^{j_1 L} I_i^{(\pi')}}{L} - \frac{\sum_{i=(j_1-1)L+1}^{jL} I_i^{(\pi)}}{L} \\
& = \frac{I_{i_1}^{(\pi')} - I_{i_1}^{(\pi)}}{L}
\end{split}
\end{equation}
Also,
\be{resolution-eq2}
0 \leq I_{i_1}^{(\pi)} \leq I_{i_1}^{(\pi')} \leq 1
\ee
\eq{resolution-eq1} and \eq{resolution-eq2} together imply that
$$
0 \leq R_{j_1}^{(\pi')} - R_{j_1}^{(\pi)} \leq \frac{1}{L}
$$
The other inequality also follows similarly or simply by observing that $\sum R_j^{(\pi)} = \sum R_j^{(\pi')} = I(W)$, by \Lref{sum-rate}.
\end{proof}

Fix an arbitrary point $\bQ = (Q_1,Q_2,\dots,Q_m)$ on the dominant face $\cD(W)$. For a given permutation $\pi$, let $\AQ \subset [m]$ be defined as follows:
\be{AQ-def}
\AQ \,\deff\, \left\{j\in [m]: R^{(\pi)}_j < Q_j\right\}
\ee
\begin{lemma}
\label{resolution2}
For any $\cB \subset \AQ$, at least one input bit from the complement set $\cB^c$ appears after some input bits of the set $\cB$ through the permutation $\pi$. 
\end{lemma}
\begin{proof}
Assume, to the contrary, that all input bits of the users in the set $\cB^c$ appear before all the input bits of the users in the set $\cB$. This together with the chain rule of mutual information and definition of $R_j^{(\pi)}$ in \eq{Rj-dfn} imply that
\be{resolution2-1}
\sum_{j \in \cB} R_j^{(\pi)} = I(U_1^L[\cB];Y,U_1^L[\cB^c]). 
\ee 
Notice that $Q$ is a point included in $\sI(W)$. Therefore, using \eq{region-def} and \eq{resolution2-1} we get
$$
\sum_{j \in \cB} Q_j \leq I(U_1^L[\cB];Y_1^L, U_1^L[\cB^c]) = \sum_{j \in \cB} R_j^{(\pi)}
$$
But since $\cB \subset \AQ$, for any $j \in \cB$, $Q_j > R_j^{(\pi)}$ which is a contradiction. This proves the lemma. 
\end{proof}

For two users $j$ and $j'$, we say that $j'$ is reachable from $j$, if there exists a sequence of users $j_1,j_2,\dots,j_t$, with $j_i \in [m]$, such that $j \rightarrow j_0 \rightarrow j_1 \rightarrow \dots \rightarrow j_t \rightarrow j'$. This sequence is also referred to as the path from $j$ to $j'$. If such path exists, then one can assume, without loss of generality, that $j_1,j_2,\dots,j_t$ are distinct. Clearly, if $j''$ is reachable from $j'$ and $j'$ is reachable from $j$, then $j''$ is reachable from $j$. 

\begin{corollary}
\label{c-path}
For any $j \in \AQ$, at least one element of the complement set $[m] - \AQ$ is reachable from $j$. 
\end{corollary}
\begin{proof}
Let $\cB$ be the set of all reachable users from $j$. Then $\cB$ is a closed set in the sense that no element in $\cB^c$ is reachable from any element of $\cB$. It implies that through the permutation $\pi$, all the input bits of $\cB$ appear after all the input bits of $\cB^c$. Then by \Lref{resolution2}, $\cB$ is not a subset of $\AQ$ which proves the corollary.   
\end{proof}

Now that we have established the necessary notations and derived desired properties of the $m$-tuple rates of the base codes, we turn to state the main theorem of this section as follows: 
\begin{theorem}
\label{cover}
Given a point $\bQ \in \cD(W)$, there exists a permutation $\pi \in \cP_L$ such that for any index $j \in [m]$, we have
$$
|Q_j - R^{(\pi)}_j| \leq \frac{m-1}{L}
$$
\end{theorem}
\begin{proof}
Let $\pi$ to be the permutation such that the Euclidean distance $\left\|\bR^{(\pi)}-\bQ\right\|$ is minimum among all the permutations in $\cP_L$ i.e. for any other permutation $\pi'$:
$$
\left\|\bR^{(\pi)}-\bQ\right\| \leq \left\|\bR^{(\pi')}-\bQ\right\|.
$$
Consider $\AQ$, defined as in \eq{AQ-def}. If $\AQ$ is empty, then it implies that for any $j \in [m]$, $Q_j \leq R^{(\pi)}_j$. But both $\bR^{(\pi)}$ and $\bQ$ have the same sum-rate i.e.
$$
\sum_{j \in [m]} R^{(\pi)}_j = \sum_{j \in [m]} Q_j = I(W).
$$
Therefore, $\bQ = \bR^{(\pi)}$ and the lemma is proved. In the case that $\AQ$ is non-empty, let $j_1$ be an arbitrary element of $\AQ$. By \Cref{c-path}, there is a path $j_1 \rightarrow j_2 \rightarrow \dots \rightarrow j_t$ with $j_i \in \AQ$, for $i = 2,\dots,t-1$, and $j_t \notin \AQ$. Without loss of generality, we can assume that $j_i$'s are distinct and hence $t \leq m$. For notational convenience, let's re-label the users so that the user $j_i$ is labeled with $i$ for $i \in [t]$. 

Define the sequence of real numbers $\left\{a_i\right\}_{i=1}^{t}$, where $a_i = Q_i - R^{(\pi)}_{i}$. We claim that $a_{i-1} - a_i \leq \frac{1}{L}$. The claim will imply that $a_1 - a_t \leq \frac{t-1}{L}$ and since $a_t$ is not a positive number, as $t \notin \AQ$, we will have 
$$
a_1 = Q_1 - R^{(\pi)}_{1} \leq a_1 - a_t \leq \frac{t-1}{L} \leq \frac{m-1}{L}
$$

By symmetry, the same set of arguments can be applied to the set $[m] - \AQ$, where one can prove that $R^{(\pi)}_{j} - Q_j \leq \frac{m-1}{L}$, for $j \in [m] - \AQ$. This will complete the proof of theorem. What remains is to prove the claim for the sequence $\left\{a_i\right\}_{i=1}^{t}$.  

For $i = 2,3,\dots,t$, let the permutation $\pi_i$ be the transposition of $\pi$ with respect to $i-1 \rightarrow i$. 
By \Lref{resolution1}, we have
\begin{align}
\label{resolution3-1}
0 \leq R_{{i-1}}^{(\pi_i)} - R_{i-1}^{(\pi)} &= R_{i}^{(\pi)} - R_{i}^{(\pi_i)} \leq \frac{1}{L}\\
\label{resolution3-2}
R_{j}^{(\pi)} &= R_{j}^{(\pi_i)}\ \text{for}\ {j \neq i-1,i}
\end{align}
On the other hand, by the choice of $\pi$, we know that
\be{resolution3-3}
\left\|\bR^{(\pi)} - \bQ\right\| \leq \left\|\bR^{(\pi_i)} - \bQ\right\| 
\ee
\eq{resolution3-2} and \eq{resolution3-3} together imply that
\be{resolution3-4}
\begin{split}
(Q_{i-1} -  R_{i-1}^{(\pi)})^2 &+ (Q_i -  R_{i}^{(\pi)})^2 \\
&\leq (Q_{i-1} -  R_{i-1}^{(\pi_i)})^2 + (Q_i -  R_{i}^{(\pi_i)})^2
\end{split}
\ee
Let $\alpha = R_{{i-1}}^{(\pi_i)} - R_{i-1}^{(\pi)}$, then by \eq{resolution3-1}, \eq{resolution3-4} can be re-written as
$$
a_{i-1}^2 + a_i^2 \leq (a_{i-1} - \alpha)^2 + (a_{i}+\alpha)^2
$$
which can be simplified as
$$
a_{i-1} - a_i \leq \alpha
$$
By \eq{resolution3-1}, $\alpha \leq \frac{1}{L}$ which completes the proof of claim.
\end{proof}

\begin{corollary}
\label{cover2}
The covering radius of the dominant face $\cD(W)$ is upper bounded by $\frac{(m-1)\sqrt{m}}{L}$.
\end{corollary}
\begin{proof}
Consider an arbitrary point $\bQ$ on the dominant face. Then the permutation $\pi$ exists as in \Tref{cover}. Then we will have
$$
\left\|\bQ - \bR^{(\pi)}\right\|^2 = \sum_{j = 1}^{m} (Q_j - R_j^{(\pi)})^2 \leq \frac{m(m-1)^2}{L^2}.
$$ 
which proves the corollary by definition of covering radius in \eq{covering-radius}.
\end{proof}

\noindent
{\bf Remark.} \looseness=-1 
In Theorem 5 of \cite{CISS}, we showed for a two-user MAC that for any point $\bQ$ on the dominant face there exists a pair of achievable rate with a distance at most $\frac{\sqrt{2}}{L}$ from $\bQ$. In other words, the covering radius for the special case of $m=2$ is upper bounded by $\frac{\sqrt{2}}{L}$. This result matches with the result of \Cref{cover2} for $m = 2$.

\section{Achieving the uniform rate region with joint decoding}
\label{sec:four}

In this section, we show that the points $\rp$ defined in \eq{Rj-dfn} are achievable by polar coding and a joint successive cancellation decoder, for all $L=2^l$ and permutations $\pi \in \cP_L$. 

Let $W$ be a given $m$-user binary-input discrete multiple access channel. Let also $n \geq l$ be a positive integer and $N = 2^n$. For $j = 1,2,\dots,m$, assume that $U_1^N[j]$ is a vector of independent and uniformly distributed bits that is generated by the $j$-th user, independent of other users. Let also $X_1^N[j] = U_1^N[j]\,G^{\otimes n}$. For $i=1,2,\dots,N$, the $m$-tuple $(X_i[1],X_i[2],\dots,X_i[m])$ is transmitted through the $i$-th independent copy of $W$ and the output is denoted by $Y_i$. This transformation together with an ordered sequence of the input bits $(U_1^N[1],U_1^N[2],\dots,U_1^N[m])$ is called a MAC polar transformation. This ordered sequence is also referred to as the decoding order, because it will specify the order in which input bits are decoded in the successive cancellation decoder. For a fixed value of $L$, we limit our attention to polarization transforms of length $N$ whose decoding orders are built upon elements of $\cP_L$, specified next.

For any permutation $\pi : [M] \rightarrow [M]$ and any $k \in \N$, where $k,M \in \N$, we define the permutation $\pi^{(k)} : [kM] \rightarrow [kM]$, as follows. First the sequence $(1,2,\dots,kM)$ is split into $k$ sub-sequences $s_1,s_2,\dots,s_M$ of length $k$ each, where $s_i = \bigl(k(i-1)+1,k(i-1)+2,\dots,ki\bigr)$. Then the permutation $\pi$ is applied to the sequence $(s_1,s_2,\dots,s_M)$ to get the sequence $s_{\pi(1)},s_{\pi(2)},\dots,s_{\pi(M)}$. At the end, the sub-sequences are expanded back to get a sequence of total length $kM$. For instance, if $M = 2$ and $\pi(1) = 2, \pi(2) =1$, then for any $k$, the permutation $\pi^{(2k)}$ replaces the first and second sub-block of length $k$ with each other. 

With the above definition, decoding orders on $(U_1^N[1],U_1^N[2],\dots,U_1^N[m])$ specified by $\pi^{(N/L)}$, for some $\pi \in \cP_L$, are considered. In that sense, the MAC polar transformations of length $L$ and decoding orders $\pi$ are the corresponding base codes. As we will see, the polar transformations of length $N$ built upon polarization base codes of length $L$ can be decomposed into $mL$ certain single-user polar transformations of length $N/L$. 

Let $W^N: \sX^{mN} \rightarrow \sY^N$ denote the channel consisting of $N$ independent copies of $W$ i.e. 
\be{Wnm}
W^N\kern-0.5pt(y^N_1|x_1^N[\cM]) 
\,\ \deff\,\
\prod_{i=1}^N W(y_i|x_i[\cM])
\vspace{-0.25ex}
\ee
where $\cM = [m]$ for notational convenience. The combined channel $\widetilde{W}_N$ is defined with transition probabilities given by
\be{Wtildem}
\widetilde{W}_N(y^N_1|u_1^N[\cM]) 
\,\ \deff\,\ W^N\kern-0.5pt(y^N_1|x_1^N[\cM]) 
\ee
where $x_1^N[j] = u_1^N[j].\Gn$, for $j \in \cM$. Let $d_1^{mN}$ denote the ordered version of the sequence $(u_1^N[1],u_1^N[2],\dots,u_1^N[m])$ according to a certain $\pi^{(N/L)}$.
The bit-channels are defined with respect to this ordered sequence, which will also enforce the decoding order in the joint successive decoding. For $j=1,2,\dots,mN$, the $j$-th bit-channel is defined as
\be{Wi-def-general}
W^{(j)}_{N}\bigl( y^N_1,d^{j-1}_1 | \hspace{1pt}d_j)\,\ \deff \,\
\frac{1}{2^{mN-1}}\hspace{-5pt}
\sum_{d_{j+1}^{mN} \in \{0,1\}^{mN-j}} \hspace{-12pt}
\widetilde{W}_N \Bigl(y^N_1\hspace{1pt}{\bigm|}\hspace{1pt}
d_1^{mN} \Bigr)
\ee
The permutation $\pi$ and consequently the decoding order are assumed to be fixed for the rest of this section. Actually the results are not specific to $\pi$ and will hold for any $\pi \in \cP_L$. 

The following lemma is similar to the Proposition 3 of \cite{Arikan} that still holds for the particular MAC polar transformation defined here and its corresponding bit-channels.
\begin{lemma}
\label{channel-split-general}
Given $N = 2^n \geq L$ independent copies of $W$ and the bit-channels defined in \eq{Wi-def-general}, for any $j$ with $1 \leq j \leq N$,
$$
W_{2N}^{(2j-1)} = W_N^{(j)} \boxcoasterisk W_N^{(j)}
$$
and
$$
W_{2N}^{(2j)} = W_N^{(j)} \circledast W_N^{(j)}
$$
\end{lemma}
\begin{proof}
Let $d^{2N}_{1,o}$ and $d^{2N}_{1,e}$ denote the even-indexed and odd-indexed subvectors, respectively. Then
\begin{align*}
W_{2N}^{(2i)}  = & \frac{1}{2^{2mN-1}} \sum_{d_{2i+1}^{2mN}} \widetilde{W}_N\bigl(y^N_{1,o} | d^{2N}_{1,o} \oplus d^{2N}_{1,e} \bigr) 
\widetilde{W}_N\bigl(y^N_{1,e} | d^{2N}_{1,e} \bigr) \\
= & \frac{1}{2} \frac{1}{2^{mN-1}} \sum_{d_{2i+1,e}^{2N}} \widetilde{W}_N\bigl(y^N_{1,e} | d^{2N}_{1,e} \bigr). \\
& \frac{1}{2^{mN-1}} \sum_{d_{2i+1,o}^N} \widetilde{W}_N\bigl(y^N_{1,o} | d^{2N}_{1,o} \oplus d^{2N}_{1,e} \bigr) \\
= & W_N^{(i)} \circledast W_N^{(i)}
\end{align*} 
where we used the definitions of bit-channels in \eq{Wi-def-general} and the channel combining operation in \eq{channel-com} together with the recursive structure of polar transformation. The other equation can be derived using the similar arguments.
\end{proof}

\Lref{channel-split-general} implies that the polar transformation of length $N$ and decoding order $\pi^{(N/L)}$ can be regarded as $n - l$ levels of polarization of the bit-channels defined with polarization base code of length $L$ and decoding order $\pi$. This is the result of following corollary:

\begin{corollary}
\label{recursion-split-general}
For a given $k \in \left[mL\right]$, let $\overline{W}$ denote the single user channel $W_{L}^{(k)}$. Let also $n \geq l$ and $N=2^n$. Then for any $\frac{(k-1)N}{L}+1 \leq j \leq \frac{kN}{L}$, 
$$
W_{N}^{(j)} = \overline{W}_{\frac{N}{L}}^{(j-\frac{(k-1)N}{L})}
$$
where the bit-channels with respect to $\overline{W}$ are defined as in \eq{Wi-def}. 
\end{corollary}
\begin{proof}
The proof is by induction on $n$. The base of induction is trivial for $n=1$. The induction steps are by the fact that channel combining recursion steps in Arikan's polar transformation, as proved in Proposition 3 of \cite{Arikan}, match with \Lref{channel-split-general}.   
\end{proof}

For $j \in [m]$, let $S^{(j)}_L$ denote the set of positions of the input bits of the $j$-th user in the ordered sequence specified by $\pi$. Then the set of indices $[mN]$ is split into $m$ subsets $S^{(j)}_N$, for $j \in [m]$, as follows:
$$
S^{(j)}_N = \left\{\, j \in [mN] ~:~ \left\lceil \frac{Lj}{N} \right\rceil \in S^{(j)}_L  \hspace{1pt}\right\}
$$
In fact, $S^{(j)}_N$ is the set of positions assigned to the $j$-th user in the decoding order specified by $\pi^{(N/L)}$. The set of good bit-channels for the $j$-th user is also defined as follows. For any $\beta \,{<}\, \shalf$ and $N = 2^n$
\begin{equation}
\label{good-def}
\cG_N^{(j)}(W,\beta)\ {\deff}\ \left\{\, j \in S^{(j)}_N ~:~ Z(W^{(j)}_{N}) < 2^{-N^{\beta}}\!\!/mN \hspace{1pt}\right\}
\end{equation}

The next theorem establishes the main result of channel polarization for the MAC polar transformation.
\begin{theorem}
\label{thm2}
For any $m$-user binary-input discrete MAC $W$, any constant $\beta \,{<}\, \shalf$ and $j \in [m]$, we have
$$
\lim_{N \to \infty} \frac{\bigl|\cG^{(j)}_N(W, \beta)\bigr|}{N} = R_j^{(\pi)} 
$$
\end{theorem}
\begin{proof}
By \Cref{recursion-split-general}, 
$$
\cG^{(j)}_N(W, \beta) = \bigcup_{i \in S^{(j)}_L} \cG_{\frac{N}{L}}(W^{(i)}_{L}, \beta)
$$ 
There exists $\beta' \,{<}\, \shalf$ such that $N^{\beta} < (\frac{N}{L})^{\beta'}$ for large enough $N$. Therefore, channel polarization theorem for the single user case \cite{AT} imply that
\begin{align*}
\lim_{N \to \infty} \frac{\bigl|\cG^{(j)}_N(W, \beta)\bigr|}{N} &\geq \frac{1}{L} \sum_{i \in S^{(j)}_L} \lim_{N \to \infty} \frac{\bigl|\cG_{\frac{N}{L}}(W^{(i)}_{L}, \beta)\bigr|}{N/L} \\
& = \frac{1}{L} \sum_{i \in S^{(j)}_L} I_i^{(\pi)} = R_j^{(\pi)}
\end{align*}
Since the fraction of all good bit-channels can not exceed the sum-rate $I(W) = \sum_{j \in [m]} R_j^{(\pi)}$, both the above inequalities should be equality. This completes the proof of theorem.
\end{proof}

\subsection{Encoding, decoding and asymptotic performance}

The encoding of the proposed scheme in the previous subsection is similar to that of original polar codes for each of the users. In fact, regardless of the choice of polarization base code the encoder for any block length $N$ is fixed, where each user multiplies a vector of length $N$ by the polarization matrix $\Gn$. The only thing that depends on the polarization base code is the indices of information bits. For a fixed length of the polarization base code $L=2^l$ and decoding order $\pi$, the set of good bit-channels for all the $m$ users $\cG_N^{(j)}(W,\beta)$ are defined in \eq{good-def}. In the polar encoder for the $j$-th user, $u_i$ is an information bit for any $i \in \cG_N^{(j)}(W,\beta)$. Otherwise, $u_i$ is frozen to a fixed value. Therefore, the rate of polar code for the $j$-th user approaches $R^{(\pi)}_j$, as $N$ goes to infinity, by \Tref{thm2}. The frozen bits are chosen uniformly at random, revealed to the receiver a priori and fixed during the course of communication. If the underlying channels are symmetric, then the frozen bits can be set to zeros.

For joint decoding, the successive cancellation decoding defined in \cite{Arikan} is performed, with some modification. The order of decoding is determined by $\pi^{(N/L)}$, defined in the previous subsection. The low-complex implementation of successive cancellation decoder invented by Ar{\i}kan in \cite{Arikan} can be extended to our scheme for MAC. The recursive steps of computing likelihood ratios can be done in a similar way using $mL$ trellises of size $N \times (n-l+1)$, corresponding to the $mL$ bit-channels defined with respect to the polarization base code of length $L$. The input likelihood ratios (LR) for these trellises can be computed using a naive way with complexity $O(2^{mL})$ for each LR. Notice that this is the worst case required complexity in the sense that depending on the choice of particular base code one can invoke the recursive structure, imposed by $G^{\otimes l}$ for each of the users, to compute the LRs more efficiently. The total decoding complexity is then upper bounded by $O(mN(n-l+1+2^{mL}))$, where $m$ and $L$ are regarded as fixed parameters in the scheme. Therefore, the decoding complexity is asymptotically $O(N \log N)$, similar to the original polar decoding.

For the asymptotic performance of the proposed polar coding, the following theorem follows similar to \cite{Arikan} and \cite{AT}.
\begin{theorem}
\label{thm-main}
For any $\beta \,{<}\, \shalf$, any $m$-user MAC $W$, any $\epsilon > 0$ and any point $\bQ$ on $\cD(W)$, there exists a family of polar codes that approaches a point on the dominant face within distance $\epsilon$ from $\bQ$. Furthermore, the probability of frame error under successive cancellation decoding is less than $2^{-N^{\beta}}$, where $N$ is the block length of the code for each of the users.
\end{theorem}
\begin{proof}
The choice of $L$ for the polarization base code depends on $\epsilon$. Fix $L = 2^l$ such that $\frac{(m-1)\sqrt{m}}{L} < \epsilon$. Then by \Cref{cover2} there exists $\pi$ such that the distance between $\bQ$ and $\rp$ is less than $\epsilon$. Fixing $L$ and $\pi$, for any block length $N = 2^n \geq L$, we construct the polar code for the $j$-th user with respect to the set of good bit-channels $\cG_N^{(j)}(W,\beta)$ defined in \eq{good-def}. Then by \Tref{thm2}, the $m$-tuple of rates approach $\rp$ as $N$ goes to infinity. The probability of frame error is bounded by the sum of the Bhattacharrya parameters of the selected bit-channels. Therefore, it is less than $2^{-N^{\beta}}$ for the joint successive cancellation decoding of all the $m$ messages together, by definition of $\cG_N^{(j)}(W,\beta)$. 
\end{proof}

\subsection{Comparison with prior works}
\label{sec:comparison}

Our approach for polarizing MAC is essentially different from those of \cite{STY,AT2}. The authors of \cite{STY,AT2} propose a framework for MAC polarization by extending the notion of channel splitting from the single user case to the two user case \cite{STY}, and then to the $m$-user case \cite{AT2}, while we propose a method to view the MAC polarization as a single user channel polarization by considering all the possible decoding orders. First, we compare the two schemes in terms of decoding complexity and then capacity achieving property.

Both of our scheme and those of \cite{STY,AT2} use a successive cancellation decoding, originally proposed by Ar{\i}kan in \cite{Arikan}. As discussed in the foregoing subsection, we have to combine the low-complex Ar{\i}kan's decoder with a basic decoder for the base code. As a result, we have an extra additive term $O(mN 2^{mL})$ in the complexity which is dominated by $O(mN (n-l+1))$, as $N$ goes to infinity, while $L$ and $m$ are assumed to be constant. On the other hand, in the scheme proposed in \cite{STY} and extended in \cite{AT2}, the likelihood of a vector of length $m$ needs to be tracked along the decoding trellis. Therefore, instead of a simple likelihood ratio, a vector of length $2^m$ for the probability of all $2^m$ possible cases has to be computed recursively. This increases the decoding complexity by a factor of $2^m$. In fact, the decoding complexity is still $O(N\log N)$, but if we look at the actual number of operations needed to complete the decoding, the decoding complexity of the scheme in \cite{STY,AT2} is $\frac{2^m}{m}$ times more than the decoding complexity of our scheme, asymptotically.  

Next, we compare the two schemes in terms of capacity-achieving property. In \Tref{thm-main}, we proved that all the points on the dominant face of the uniform rate region can be achieved with arbitrary small resolution. As pointed out before, the scheme proposed in \cite{STY,AT2} does not necessarily achieve the whole uniform rate region. It is only guaranteed that one point on the dominant face is achievable. Here, we provide an example for a two user MAC such that the scheme proposed in \cite{STY}, achieves only one point on the dominant face. The example can be extended to the case that the number of users is a power of two.

Let $W': \left\{0,1\right\} \rightarrow \sY'$ be an arbitrary single user binary-input DMC. Then the two user multiple access channel $W: \left\{0,1\right\} \rightarrow \sY$ is constructed as follows. Let $\sY = \sY' \times \sY'$. For any two binary inputs $u$ and $v$ corresponding to the first and second user, $u \oplus v$ and $v$ are transmitted through two independent copies of $W'$ with outputs $y_1$ and $y_2$. Then the output of $W$ is defined to be $y = (y_1,y_2) \in \sY$. The diagram of $W$ is depicted in Figure\,\ref{MAC-example}. 

\begin{figure}[h]
\centering
\includegraphics[width=0.7\linewidth]{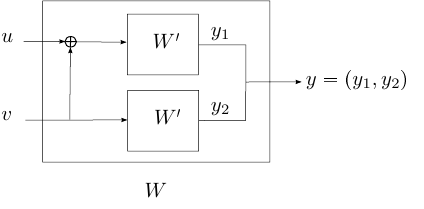}
\caption{A constructed example for two user MAC}
\label{MAC-example}
\end{figure} 

Let $I = I(W')$, $I^- = I(W^{'-})$ and $I^+ = I(W^{'+})$. Then the uniform rate region of $W$ is as shown in Figure\,\ref{MAC-example-region}. 

\begin{figure}[h]
\centering
\includegraphics[width=0.7\linewidth]{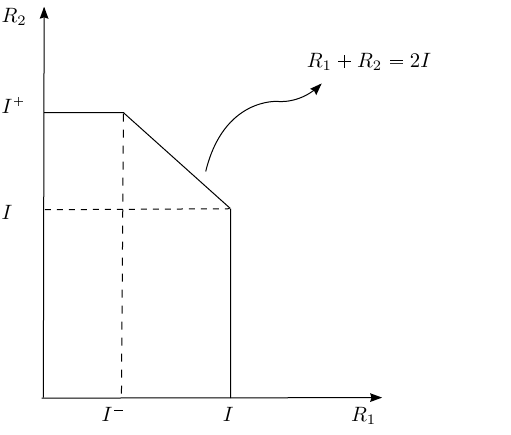}
\caption{The uniform rate region for the constructed example}
\label{MAC-example-region}
\end{figure} 

It is easy to observe that a polar transformation of length $N$ on two user MAC $W$, in the sense defined in \cite{STY}, is equivalent to a single user polar transformation of length $2N$ over $W'$. Let $U_1^N$ and $V_1^N$ be the input bits by the first and the second user, respectively. Let $W^{(i)}_N$ be the $i$-th two user bit-channel observed by $(U_i,V_i)$. Then
$$
I^{(2)}(W^{(i)}_N) = I(V_i;Y_1^i,U_1^i,V_1^{i-1}) = I(W'^{(2i)}_{2N})
$$
Also,
\begin{align*}
I^{(1)}(W^{(i)}_N) & = I(U_i;Y_1^i,U_1^{i-1},V_1^i) \\
& \geq I(U_i;Y_1^i,U_1^{i-1},V_1^{i-1}) = I(W'^{(2i-1)}_{2N})
\end{align*}
And,
$$
I(W^{(i)}_N) =  I(W'^{(2i)}_{2N}) + I(W'^{(2i-1)}_{2N})
$$
Notice that as $N$ goes to infinity, $W'^{(2i-1)}_{2N}$ and $W'^{(2i)}_{2N}$ are either both good or both not good, for almost all $i$'s. In fact, the fraction of $i$'s for which only one of them is good approaches zero. Therefore, the triple $\bigl(I^{(1)}(W^{(i)}_N), I^{(2)}(W^{(i)}_N), I(W^{(i)}_N)\bigr)$ approaches either $(0,0,0)$ or $(1,1,2)$ among the five possible cases discussed in \cite{STY}. The point $(0,0,0)$ corresponds to the case that both of the inputs needs to be frozen to fixed values. The point $(1,1,2)$ corresponds to the case that both the inputs carry information bits. Thus, the only achievable point on the dominant face by this scheme is the point $(I,I)$, and the achievable region is the square with vertices $(0,0),(0,I),(I,0),(I,I)$. On the other hand, as proved in this section, the whole region shown in Figure\,\ref{MAC-example-region} can be achieved by our proposed scheme with joint successive decoder.

\section{Discussions and conclusion}
\label{sec:six}

In this paper, we considered the problem of designing polar codes for transmission over general $m$-user multiple access channels. The key observation behind our work is to view the polar transformations for the $m$-users across the multiple access channel as a unified MAC polar transformation. We showed that the MAC polar transformation with a certain decoding order can be split into $mL$ single user polar transformations, where $L$ is the length of the base code. We also proved that the covering radius of the dominant face is upper bounded by $\frac{(m-1)\sqrt{m}}{L}$. Therefore, by letting $L$ to grow large, we are able to achieve the entire uniform rate region with the constructed MAC polar code.

The idea behind the joint successive cancellation decoding is that decoding the messages of different users successively is essentially following the same rule as in the successive cancellation decoder of polar codes. Therefore, one can invoke the joint successive cancellation decoder with a certain decoding order that is the result of shuffling the input bits of all users. We also derived the direct relation between the decoding order and the polarization base codes in order to guarantee channel polarization. We proved that by shuffling the order of the information bits of different users to be decoded, which will also enforce the choice of MAC polarization base code, the achievability of the entire uniform rate region is guaranteed. 

Since the individual codes for each of the users can be regarded as single user polar codes, all the existing methods for improving the performance of single user polar codes can be applied on top of our proposed MAC-polar codes. For instance, the individual polar codes can be made systematic as suggested by Ar{\i}kan \cite{Arikan4} in order to improve the bit error rate. Also, concatenation of polar codes with other block codes \cite{RS-polar} or list decoding of polar codes \cite{TV} can be used to boost the finite-length performance. Also, as discussed in \cite{CISS}, one can invoke the compound polar codes proposed in \cite{compound} in order to improve the finite length performance if time sharing is used.

\bibliographystyle{IEEEtran}
\bibliography{polar}

\end{document}